\documentclass[11pt]{article}
 
\usepackage{fullpage,amsmath,amssymb,url,hyperref}
\usepackage{graphicx}

\def\QED{\ensuremath{{\square}}}
\def\markatright#1{\leavevmode\unskip\nobreak\quad\hspace*{\fill}{#1}}

\def\eB{\mathrm{eB}}
\def\KL{\mathrm{KL}}

\def\du{\mathrm{d}u}

\def\calH{\mathcal{H}}
\def\calL{\mathcal{L}}
\def\calX{\mathcal{X}}
\def\calF{\mathcal{F}}
\def\calP{\mathcal{P}}
\def\calQ{\mathcal{Q}}
\def\calC{\mathcal{C}}

\def\dmu{\mathrm{d}\mu}
\def\supp{\mathrm{supp}}

\def\leftsup#1{{}^{#1}}

\def\eJ{\mathrm{eJ}}

\def\qccvJ{\leftsup{\mathrm{qccv}}J}
\def\qcvxJ{\leftsup{\mathrm{qcvx}}J}
\def\qcvxB{\leftsup{\mathrm{qcvx}}B}

\def\eqdef{:=}
\def\bbR{\mathbb{R}}

\newtheorem{definition}{Definition}
\newtheorem{theorem}{Theorem}
\newtheorem{property}{Property}
\newtheorem{remark}{Remark}

\newenvironment{proof}
 {\begin{trivlist}\item[\hskip\labelsep{\bf Proof.}]}
 {\markatright{\QED}\end{trivlist}}

\title{A note on the quasiconvex Jensen divergences and the quasiconvex Bregman divergences derived thereof}

\author{
Frank Nielsen\footnote{E-mail: {\tt Frank.Nielsen@acm.org}. Web: \url{https://FrankNielsen.github.io/}} \and 
Ga\"etan Hadjeres\footnote{E-mail: {\tt Gaetan.Hadjeres@sony.com}}
}

\date{$^*$ Sony Computer Science Laboratories Inc., Tokyo, Japan\\
$^\dagger$ Sony Computer Science Laboratories, Paris, France}

\begin{document}
\maketitle

\begin{abstract}
We first introduce the class of strictly quasiconvex and strictly quasiconcave Jensen divergences which are oriented (asymmetric) distances, and study some of their properties.
We then define the strictly quasiconvex Bregman divergences as the limit case of  scaled and skewed quasiconvex Jensen divergences, and report a simple closed-form formula which shows that these divergences are only  pseudo-divergences at countably many inflection points of the generators. 
To remedy this problem, we propose the $\delta$-averaged quasiconvex Bregman divergences which integrate the pseudo-divergences over a small neighborhood in order obtain a proper divergence. The formula of $\delta$-averaged quasiconvex Bregman divergences extend even to non-differentiable strictly quasiconvex generators.
These quasiconvex Bregman divergences between distinct elements have the property to always have  one orientation finite while the other orientation is infinite.
We show that these quasiconvex Bregman divergences can also be interpreted as limit cases of generalized skewed Jensen divergences with respect to comparative convexity by using power means.
Finally, we illustrate how these quasiconvex Bregman divergences naturally appear as  equivalent divergences for the Kullback-Leibler divergences between probability densities belonging to a same parametric family of distributions with nested supports.
\end{abstract}

quasiconvexity, inflection points, comparative convexity, power means, nested densities.

%%%%
\section{Introduction, motivation, and contributions}
%%%%%%%
A {\em dissimilarity} $D(O,O')$ is a measure of the deviation of an object $O'$ from a reference object $O$ (i.e., $D_O(O')\eqdef D(O,O')$) which
 satisfies the following two basic properties:
\begin{description}
	\item[Non-negativity.] $D(O,O')\geq 0, \forall O,O'$
	\item[Law of the indiscernibles.] $D(O,O')=0$ if and only if $O=O'$.
\end{description}
In other words, a dissimilarity $D(O,O')$ satisfies $D(O,O')\geq 0$ with equality if and only if $O=O'$.
A {\em pseudo-dissimilarity} is a measure of deviation for which the
non-negativity property holds but not necessarily the law of the indiscernibles~\cite{HolderDiv-2017}.
The objects can be vectors, probability distributions, random variables, strings, graphs, etc.
In general, a dissimilarity may not be symmetric, i.e., potentially we may have $D(O,O')\not=D(O',O)$.
In that case, the dissimilarity is said to be {\em oriented}, and we consider the following two reference orientations of the dissimilarity: 
the {\em forward ordinary dissimilarity} $D(O:O')$ and its associated {\em reverse dissimilarity} $D^r(O:O')\eqdef D(O':O)$. 
Notice that we used the ':' notation instead of the comma delimiter ',' between the dissimilarity arguments to emphasize that the dissimilarity may be asymmetric.
In the literature, a dissimilarity is also commonly called {\em a divergence}~\cite{IG-2016} although several additional meanings may be associated to this term like a dissimilarity between {\em probability distributions} instead of vectors (e.g., the Kullback-Leibler divergence~\cite{CT-2012} in information theory) or like a notion of smoothness (e.g., a $C^3$ contrast function in information geometry~\cite{IG-2016}). A dissimilarity may also be loosely called a {\em distance} although this may convey to mathematicians in some contexts the additional notion of a dissimilarity satisfying the metric axioms (non-negativity, law of the indiscernibles, symmetry and triangular inequality).

The {\em Bregman divergences}~\cite{Bregman-1965,Bregman-1967} were introduced in operations research, and are widely used nowadays in machine learning and information sciences. For a strictly convex and smooth generator $F$, called the {\em Bregman generator}, we define the corresponding Bregman divergence between parameter vectors $\theta$ and $\theta'$ as:
\begin{equation}
B_F(\theta:\theta')= F(\theta)-F(\theta')-(\theta-\theta')^\top \nabla F(\theta').
\end{equation}
Bregman divergences are always finite, and generalize many common distances~\cite{BD-2005}, including the Kullback-Leibler (KL) divergence and the squared Euclidean and Mahalanobis distances. 
Furthermore, the KL divergence between two probability densities belonging to a same exponential family~\cite{EF-Barndorff-2014,BD-2005}  amount to a {\em reverse Bregman divergence} between the corresponding parameters when setting the Bregman generator to be  the cumulant function of the exponential family~\cite{KLEFBD-2001}.
Moreover, a bijection between regular exponential families~\cite{EF-Barndorff-2014} and the so-called class of ``regular Bregman divergences'' was reported in~\cite{BD-2005} and used for learning statistical mixtures showing that the expectation-maximization algorithm is equivalent to a Bregman soft clustering algorithm.
Bregman divergences have been extended to many non-vector data types like matrix arguments~\cite{MatrixBD-2013} or functional arguments~\cite{FunctionalBD-2008}.

In this note, we consider defining the notion of Jensen divergences~\cite{BR-2011} for strictly quasiconvex or strictly quasiconcave generators, and the induced notion of Bregman divergences. 
We term them {\em quasiconvex Bregman divergences} (and omit to prefix it by 'strictly' for sake of brevity). 
We then establish a connection between the KL divergence between parametric families of densities with nested supports and these quasiconvex Bregman divergences.

We summarize our main contributions as follows:

\begin{itemize}
	\item By using quasiconvex generators instead of convex generators, 
	we define  the skewed quasiconvex Jensen  divergences (Definition~\ref{def:qcvxJ}) and derived thereof quasiconvex Bregman divergences (Definition~\ref{def:qcvxB} and Theorem~\ref{thm:qcvxB}).
	The quasiconvex Bregman divergences turn out to be only pseudo-divergences at  inflection points of the generator.
	Since this happens only at countably many points, we still loosely call them quasiconvex Bregman divergences.
	We can also integrate the quasiconvex Bregman (pseudo-)divergence over a small neighborhood  and obtain a
	 $\delta$-averaged quasiconvex Bregman divergence in~\S\ref{sec:deltaaveraged}. 
	The $\delta$-averaged quasiconvex Bregman divergence are also well-defined for strictly quasiconvex but not differentiable generators.
	Quasiconvex Bregman divergences between distinct parameters always have one orientation finite while the other one evaluates to infinity.
 
	\item We show that quasiconvex Jensen divergences and quasiconvex Bregman divergences can be reinterpreted  as generalized Jensen and Bregman divergences with {\em comparative convexity}~\cite{ComparativeConvexity-2018,CC-BD-2017} using power means in the limit case (\S\ref{sec:PowerMeanJ} and \S\ref{sec:PowerMeanJ}).
	
	\item We exhibit some parametric families of probability distributions with strictly nested supports such that the Kullback-Leibler divergences between them amount to  equivalent quasiconvex Bregman divergences (\S\ref{sec:StatParamDiv}).
	
\end{itemize}

The paper is organized as follows:
Section~\ref{sec:qvxJ} defines the quasiconvex and quasiconcave difference distances by analogy to Jensen difference distances~\cite{RaoBD-1985,BR-2011}, study some of their properties, and show how to obtain them as generalized Jensen divergences~\cite{CC-BD-2017} obtained from comparative convexity using power means. Henceforth their name: {\em quasiconvex Jensen divergences}. When the generator is quasilinear instead of quasiconvex, we call them {\em quasilinear Jensen divergences}.
We then define the quasiconvex Bregman divergences in \S\ref{sec:qvxBD} as limit cases of scaled and skewed quasiconvex Jensen divergences, and report a closed-form formula which highlights the fact that one orientation of the distance is always finite while the other one is always infinite (for divergences between distinct elements). 
Since the quasiconvex Bregman divergences are only pseudo-divergences at inflection points, we define the $\delta$-averaged quasiconvex Bregman divergences in \S\ref{sec:deltaaveraged}.
We also recover the formula by taking the limit case of power means Bregman divergences that were introduced using comparative convexity~\cite{CC-BD-2017}.

In \S\ref{sec:StatParamDiv}, we consider the problem of finding parametric family of probability distributions for which the Kullback-Leibler divergence amount to a quasiconvex Bregman divergence. We illustrate one example showing that nested supports of the densities ensure the property of having one orientation finite while the other one is infinite.
Finally, \S\ref{sec:concl} concludes this note and hints at applications perspectives of these quasiconvex Bregman divergences, including flat and hierarchical clustering.

%%%%%%
\section{Divergences based on inequality gaps of quasiconvex or quasiconcave generators}\label{sec:qvxJ}
%%%%%%

\subsection{Quasiconvex and quasiconcave difference dissimilarities}

In this work, a divergence or distance $D(\theta:\theta')$ refers to a dissimilarity such that $D(\theta:\theta')\geq 0$ 
with equality iff. $\theta=\theta'$. A pseudo-divergence or pseudo-distance only satisfies the non-negativity property but not necessarily the law of the indiscernibles of the dissimilarities.

Consider a function $Q:\Theta\subset\bbR^D\rightarrow \bbR$ which satisfies the following ``Jensen-type'' inequality~\cite{Boyd-2004} for any $\alpha\in (0,1)$:
\begin{equation}\label{eq:qgap}
Q((\theta\theta')_{\alpha}) < \max\{Q(\theta),Q(\theta')\}, \quad \theta\not=\theta'\in\Theta\subset\bbR,
\end{equation}
where $(\theta\theta')_{\alpha}\eqdef (1-\alpha)\theta+\alpha\theta'$ denotes the {\em weighted linear interpolation} of $\theta$ with $\theta'$, and $\Theta$ the parameter space. Function $Q$ is said {\em strictly quasiconvex}~\cite{Quasiconvex-1971,Bereanu-1972,QuasiconvexAnalysis-2007,Boyd-2004} as it relaxes the strict convexity inequality: 
\begin{equation}
Q((\theta\theta')_{\alpha}) < (1-\alpha)Q(\theta)+\alpha Q(\theta')\leq \max\{Q(\theta),Q(\theta')\}.
\end{equation}
Let $\calQ$ denote the space of such {\em strictly quasiconvex real-valued function},
and 
let $\calC$ denote the space of strictly convex functions.
We have  $\calC\subset\calQ$:
Any strictly convex function or any strictly increasing function is quasiconvex, but not necessarily the converse:
Some examples of quasiconvex functions which are not convex are $Q(\theta)=\sqrt{\theta}$, $Q(\theta)=\theta^3$, $Q(\theta,\theta')=\log(\theta^2+(\theta')^2)$, etc.
Decreasing and then increasing functions are quasiconvex but may not be necessarily smooth.
Some concave functions like $Q(\theta)=\log \theta$ are quasiconvex.
The sum of quasiconvex functions are not necessarily quasiconvex.
In the same spirit that function convexity can be reduced to set convexity via the epigraph representation of the function,
a function $Q$ is quasiconvex if the {\em level set} $L_\alpha\eqdef \{x: Q(x)\leq \alpha \}$ is (set) convex for all $\alpha\in\bbR$.
When $Q$ is univariate, a quasiconvex function is also commonly called {\em unimodal} (i.e., decreasing and then increasing function). 
Thus a multivariate quasiconvex function can be characterized as being unimodal along each line of its domain.
Figure~\ref{fig:exqcv} displays some examples of quasiconvex functions with one function that fails to be quasiconvex.
%Furthermore, let us mention that a continuous density $p(x)$ of $\bbR^d$ is {\em unimodal} if $-\log p(x)$ is quasiconvex and strongly unimodal when $-\log p(x)$ is convex  (see~\cite{EF-Barndorff-2014},\S 6.2).
Notice that strictly monotonic functions which are {\em both} strictly quasiconvex and strictly quasiconcave  are termed {\em strictly quasilinear}. 
The ceil function $\mathrm{ceil}(\theta)=\inf\{z\in\mathbb{Z}\ :\ z\geq\theta\}$ is an example of quasilinear function (idem for the floor function).
Another example, are the linear fractional functions $Q_{a,b,c,d}(\theta)=\frac{a^\top \theta+b}{c^\top \theta+d}$ which are quasilinear functions on the domain $\Theta=\{\theta\ :\ c^\top\theta+d>0\}$.
We denote by $\calL\subset\calQ$ the set of strictly quasilinear functions, and by $\calH$ the  set of strictly quasiconcave functions.

\begin{figure}%
\centering
\includegraphics[width=0.85\columnwidth]{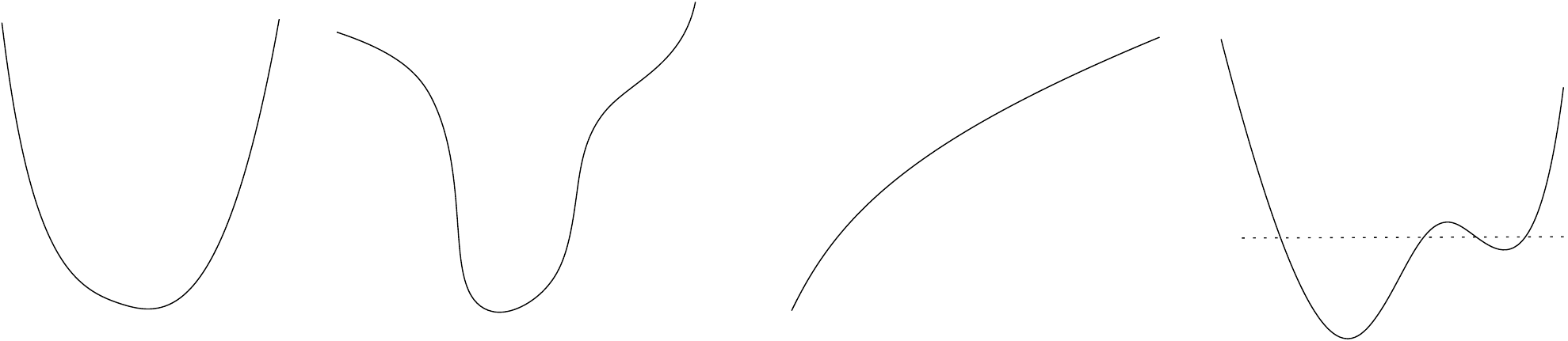}%

\caption{The first three functions (from left to right) are quasiconvex because any level set is convex, but the last function is not quasiconvex because the dotted line intersects the function in four points (and therefore the level set is not convex). The first function is convex, the second function is quasiconvex but not convex (a chord may intersect the function in more than two points), the third function is monotonous and here concave (quasilinear)).}%
\label{fig:exqcv}%
\end{figure}

\begin{definition}[Quasiconvex difference distance]\label{def:qcvxJ}
The {\em quasiconvex difference distance} (or qcvx distance for short) for $\alpha\in(0,1)$ is defined as the inequality difference gap of Eq.~\ref{eq:qgap}
\begin{eqnarray}\label{eq:qvexdist}  
\qcvxJ_Q^\alpha(\theta:\theta') &\eqdef& \max\{Q(\theta),Q(\theta')\} - Q((\theta\theta')_{\alpha}) \geq 0,\\
&=& \max\{Q(\theta),Q(\theta')\} - Q((1-\alpha)\theta+\alpha\theta')).
\end{eqnarray}
\end{definition}

By definition, the quasiconvex difference distance is a dissimilarity satisfying $\qcvxJ_Q^\alpha(\theta:\theta')=0$ iff. $\theta=\theta'$
 when the generator $Q$ is {\em strictly} quasiconvex (see Eq.~\ref{eq:qgap}).

\begin{remark}
Notice that we could also have defined a {\em log-ratio gap}~\cite{HolderDiv-2017} as a dissimilarity:
\begin{eqnarray}
\qcvxJ{}L_Q^\alpha(\theta:\theta') \eqdef -\log \left(\frac{Q((\theta\theta')_{\alpha})}{\max\{Q(\theta),Q(\theta')\}}\right).
\end{eqnarray}
However, in that case we should have required the extra condition that the generator does not vanish in the domain, i.e., 
$Q(\theta)\not=0$ for any $\theta\in\Theta$.
\end{remark}

\begin{property}
Let $a>0$ and $b\in\bbR$, and define $Q_{a,b}(\theta)=aQ(\theta)+b$. Functions $Q_{a,b}$ are quasiconvex, and
$\qcvxJ_{Q_{a,b}}^\alpha(\theta:\theta')=a\ \qcvxJ_Q^\alpha(\theta:\theta')$.
\end{property}

Similarly, we can characterize a {\em strictly quasiconcave real-valued function} $H\in\calH:\Theta\subset\bbR^D\rightarrow \bbR$ 
by the following inequality for $\alpha\in(0,1)$:
\begin{equation}
H((\theta\theta')_{\alpha}) > \min\{H(\theta),H(\theta')\}, \quad \theta\not=\theta'\in\Theta\subset\bbR^D.
\end{equation}
This allows one to define the {\em quasiconcave difference distance} (or qccv distance for short):

\begin{definition}[Quasiconcave difference distance]
For $Q$ a quasiconcave function and $\alpha\in(0,1)$, we define the quasiconcave distance as:
\begin{eqnarray}
\qccvJ_H^\alpha(\theta:\theta') &\eqdef&   H((\theta\theta')_{\alpha})- \min\{H(\theta),H(\theta')\},\\
&=& H((1-\alpha)\theta+\alpha\theta')- \min\{H(\theta),H(\theta')\}
\end{eqnarray}
\end{definition}
Similarly, we have $\qccvJ_{H_{a,b}}^\alpha(\theta:\theta')=a\ \qccvJ_H^\alpha(\theta:\theta')$
for $a>0$ and $b\in\bbR$.

Now, observe that for any $a,b\in\bbR$, we have\footnote{Indeed, $\max\{a,b\}=\frac{a+b}{2}+\frac{1}{2}|b-a|=-(\frac{-a-b}{2}-\frac{1}{2}|b-a|)=-(\frac{-a-b}{2}-\frac{1}{2}|-b+a|)=-\min\{-a,-b\}$.} $\min\{a,b\}=-\max\{-a,-b\}$ (or equivalently $\max\{a,b\}=-\min\{-a,-b\}$).
Thus it follows the following identity:

\begin{property}
A quasiconcave difference distance with quasiconcave generator $H$ is equivalent to a quasiconvex difference distance for the quasiconvevx generator $Q=-H$:

\begin{equation}
\qccvJ_H^\alpha(\theta:\theta')=\qcvxJ_{-H}^\alpha(\theta:\theta'),\quad \qcvxJ_{Q}^\alpha(\theta:\theta')=\qccvJ_{-Q}^\alpha(\theta:\theta').
\end{equation}
\end{property}

\begin{proof}
\begin{eqnarray}
\qccvJ_H^\alpha(\theta:\theta') &=&   H((\theta\theta')_{\alpha})- \min\{H(\theta),H(\theta')\},\\
&=&  \max\{-H(\theta),-H(\theta')\}-(-H((\theta\theta')_{\alpha})),\\
&=& \qcvxJ_{-H}^\alpha(\theta:\theta').
\end{eqnarray}
\end{proof}

Therefore, we consider without loss of generality quasiconvex difference distances in the reminder.

%%%
\subsection{Relationship of quasiconvex difference distances with Jensen difference distances}
%%%%

Since for any $a, b\in\bbR$, we have $\max(a,b)=\frac{a+b}{2}+\frac{1}{2}|b-a|$, $\min(a,b)=\frac{a+b}{2}-\frac{1}{2}|b-a|$ and 
$\max(a,b)-\min(a,b)=|b-a|$, we can rewrite Eq.~\ref{eq:qvexdist} to get
\begin{eqnarray}
\qcvxJ_Q^\alpha(\theta:\theta') &=& \frac{Q(\theta)+Q(\theta')}{2}+\frac{1}{2} \left|Q(\theta)-Q(\theta')\right|- Q((\theta\theta')_{\alpha}),\\
&=& \eJ_Q^\alpha(\theta:\theta')+\frac{1}{2} \left|Q(\theta)-Q(\theta')\right|+Q(\theta)\left(\alpha-\frac{1}{2}\right)+
Q(\theta')\left(\frac{1}{2}-\alpha\right),
\end{eqnarray}
where
\begin{eqnarray}
\eJ_Q^\alpha(\theta,\theta')\eqdef (Q(\theta)Q(\theta'))_{\alpha}-Q\left((\theta\theta')_\alpha\right),
\end{eqnarray}
is called the {\em extended Jensen divergence}, a Jensen-type divergence {\em extended} to quasiconvex generators instead of ordinary convex generators.  

\begin{property}[Upperbounded the extended Jensen divergence  by $\qcvxJ_Q^\alpha$]
We have:
\begin{eqnarray}
\eJ_Q^\alpha(\theta:\theta') \leq \qcvxJ_Q^\alpha(\theta:\theta')
\end{eqnarray}
since $(Q(\theta)Q(\theta'))_{\alpha}\leq \max\{Q(\theta),Q(\theta')\}$.
In particular, when $Q=F$ is strictly convex, we have $0\leq J_F^\alpha(\theta:\theta')\leq \qcvxJ_F^\alpha(\theta:\theta')$.
\end{property}

Notice that $\eJ_Q^\alpha(\theta,\theta')\geq 0$ when $Q$ is strictly convex, but may be negative when only quasiconvex.
For example, $Q(\theta)=\log\theta$ is a quasiconvex and concave function, and therefore $\eJ_Q^\alpha(\theta,\theta')\leq 0$.

When $\alpha=\frac{1}{2}$, we get the following identity:

\begin{property}[Regularization of extended Jensen divergences]
\begin{eqnarray}
\qcvxJ_Q(\theta:\theta') &=& \frac{Q(\theta)+Q(\theta')}{2}+\frac{1}{2}|Q(\theta)-Q(\theta')|- Q\left(\frac{\theta+\theta'}{2}\right),\\
&=& \eJ_Q(\theta,\theta')+\frac{1}{2}|Q(\theta)-Q(\theta')|,
\end{eqnarray}
where
\begin{eqnarray}
\eJ_Q(\theta,\theta')\eqdef \frac{Q(\theta)+Q(\theta')}{2}-Q\left(\frac{\theta+\theta'}{2}\right),
\end{eqnarray}
is an {\em extension} of the Jensen divergence~\cite{BR-1982,RaoBD-1985} to a quasiconvex generator $Q$.
\end{property}

Thus when the generator is convex, we can interpret the quasiconvex divergence as a $\ell_1$-regularization of the ordinary Jensen divergence.
When the generator $Q$ is not convex, beware that $\eJ_Q(\theta,\theta')$ may be negative but we always have
$\eJ_Q(\theta,\theta')\geq -\frac{1}{2}|Q(\theta)-Q(\theta')|$.

Similarly, when the generator $H$ is strictly quasiconcave, we rewrite the quasiconvex difference distance as
\begin{eqnarray}
\qccvJ_H(\theta:\theta') &=& H\left(\frac{\theta+\theta'}{2}\right) - \frac{H(\theta)+H(\theta')}{2}+\frac{1}{2}|H(\theta)-H(\theta')|,\\
&=& \eJ_{-H}(\theta,\theta')+\frac{1}{2}|H(\theta)-H(\theta')|.
\end{eqnarray}

%%%%%
\subsection{Quasiconvex difference distances: The viewpoint of comparative convexity}\label{sec:PowerMeanJ}
%%%%
In~\cite{CC-BD-2017}, a generalization of the skewed Jensen divergences with respect to comparative convexity~\cite{ComparativeConvexity-2018} is obtained using a pair of weighted means.
A {\em mean}  between two reals $x$ and $y$ belonging to an interval $I\subset\bbR$ is a bivariate function $M(x,y)$ such that
\begin{equation}
\min\{x,y\} \leq M(x,y)  \leq \max\{x,y\}. 
\end{equation}
That is, a mean satisfies the {\em in-betweeness property} (see~\cite{ComparativeConvexity-2018}, p. 328).
A {\em weighted mean} $M_\alpha$ for $\alpha\in [0,1]$ can always be built from a mean by using the dyadic expansion of real numbers, see~\cite{ComparativeConvexity-2018}.

Consider two {\em weighted means} $M_\alpha$ and $N_\alpha$. 

A function $F$ is said {\em $(M,N)$} convex iff:
\begin{equation}
N_\alpha(F(\theta),F(\theta')) \geq F(M_\alpha(\theta,\theta')),\quad \theta,\theta'\in\Theta.
\end{equation}
We recover the ordinary convexity when $M_\alpha=N_\alpha=A_\alpha$, where $A_\alpha(x,y)=(1-\alpha)x+\alpha y$ is the weighted arithmetic mean.

We can define the $\alpha$-skewed {\em $(M,N)$-Jensen divergence} as:
\begin{equation}\label{eq:sJCCD}
J_{F,\alpha}^{M,N}(\theta:\theta') \eqdef N_\alpha(F(\theta),F(\theta'))-F(M_\alpha(\theta,\theta')).
\end{equation}
By definition, $J_{F,\alpha}^{M,N}(\theta:\theta')\geq 0$ when $F$ is a $(M,N)$-strictly convex function.

A {\em quasi-arithmetic mean}~\cite{ComparativeConvexity-2018} is defined  for a continuous strictly increasing  function $f:I\subset\bbR \rightarrow J\subset\bbR$ as:
\begin{equation}
M_f(p,q) \eqdef f^{-1}\left( \frac{f(p)+f(q)}{2}  \right).
\end{equation}
These quasi-arithmetic means are  also called  Kolmogorov-Nagumo-de Finetti means~\cite{Kolmogorov-1930,Nagumo-1930,deFinetti-1931}.
Without loss of generality, we assume strictly increasing functions instead of monotonic functions since $M_{-f}=M_f$.
By choosing $f(x)=x$, $f(x)=\log x$ or $f(x)=\frac{1}{x}$,  we recover the Pythagorean arithmetic, geometric, and harmonic means, respectively.

Now, consider the family of {\em power means} for $x,y>0$:
\begin{equation}
P_0(x,y)\eqdef \sqrt{xy},\quad P_\delta(x,y) \eqdef \left(\frac{x^\delta+y^\delta}{2}\right)^{\frac{1}{\delta}},\quad \delta\not=0.
\end{equation}
These means  fall in the class of quasi-arithmetic means 
obtained  for $f_\delta(x)=x^\delta$  for $\delta\not =0$ with  $I=J=(0,\infty)$, and include in the limit cases the maximum and minimum values:
$\lim_{\delta\rightarrow +\infty} P_\delta(a,b)=\max\{a,b\}$ and $\lim_{\delta\rightarrow -\infty} P_\delta(a,b)=\min\{a,b\}$.

The {\em power mean Jensen divergence}~\cite{CC-BD-2017} is defined as a special case of the $(M,N)$-Jensen divergence by:
\begin{equation}
J_{F}^{P_\delta}(\theta:\theta') \eqdef J_{F}^{A,P_\delta}(\theta:\theta') = P_\delta(F(\theta),F(\theta'))-F((\theta\theta')_\alpha),
\end{equation}
for a $(A,P_\delta)$ strictly convex generator $F$.

Let us now observe that the quasiconvex difference distance is a {\em limit case} of power mean Jensen divergences:
 
\begin{property}[$\qcvxJ_Q$ as a limit case of power mean Jensen divergences]
We have
\begin{equation}
\qcvxJ_Q(\theta:\theta')=\lim_{\delta\rightarrow\infty} J_{F}^{P_\delta}(\theta:\theta').
\end{equation}
\end{property}

Notice that a strictly quasiconvex function $Q$ is interpreted as a $(A,\max)$-strictly convex function in comparative convexity, a limit case of $(A,P_\delta)$-convexity.
From now on, we term the quasiconvex difference distance the {\em quasiconvex Jensen divergence}.

%%%%
\section{Bregman divergences for quasiconvex generators}\label{sec:qvxBD}
%%%%%

\subsection{Quasiconvex Bregman divergences as limit cases of quasiconvex Jensen divergences}

Recall that for a strictly quasiconvex generator $Q$, define the {\em $\alpha$-skewed quasiconvex distance} for $\alpha\in(0,1)$ as
\begin{eqnarray}
\qcvxJ_Q^\alpha(\theta:\theta') \eqdef \max\{Q(\theta),Q(\theta')\} - Q((\theta\theta')_{\alpha}).
\end{eqnarray}
We have 
\begin{equation}
\qcvxJ_Q^\alpha(\theta:\theta')\geq 0,
\end{equation}
with equality if and only if $\theta=\theta'$.
Notice that we do not require smoothness~\cite{LB-2012} of $Q$, and $\qcvxJ_Q=\qcvxJ_Q^{\frac{1}{2}}$ is symmetric.
For an asymmetric divergence $D(\theta:\theta')$, denote $D^r(\theta:\theta')=D(\theta':\theta)$ the {\em reverse divergence}.

By analogy to Bregman divergences~\cite{BD-2005} being interpreted as {\em limit cases} of scaled and skewed Jensen divergences~\cite{Zhang-2004,BR-2011}:
\begin{eqnarray}
\lim_{\alpha\rightarrow 1^-} \frac{1}{\alpha(1-\alpha)} J_F^\alpha(\theta:\theta')&=& B_F(\theta:\theta'),\\
\lim_{\alpha\rightarrow 0^+} \frac{1}{\alpha(1-\alpha)} J_F^\alpha(\theta:\theta')&=& B_F^r(\theta:\theta')=B_F(\theta':\theta).
\end{eqnarray}

Let us  define the following divergence:

\begin{definition}[Quasiconvex Bregman pseudo-divergence]\label{def:qcvxB}
For a strictly quasiconvex generator $Q\in\calQ$, we define the quasiconvex Bregman pseudo-divergence as 
\begin{equation}\label{eq:qcvxB}
\boxed{\qcvxB_Q(\theta:\theta') \eqdef \lim_{\alpha\rightarrow 1^-} \frac{1}{\alpha(1-\alpha)}\qcvxJ_Q^\alpha(\theta:\theta')}.
\end{equation}
\end{definition}

As it will be shown below, we get only a pseudo-divergence in the limit case.

\begin{theorem}[Formula for the  quasiconvex Bregman pseudo-divergence]\label{thm:qcvxB}
For a strictly quasiconvex and differentiable generator $Q$, the quasiconvex Bregman pseudo-divergence is
\begin{equation}
\boxed{
\qcvxB_Q(\theta:\theta') =\left\{
\begin{array}{ll}
-(\theta-\theta')^\top \nabla Q(\theta') & \mbox{if $Q(\theta)\leq Q(\theta')$}\\
+\infty & \mbox{otherwise (i.e., $Q(\theta)> Q(\theta')$)}.
\end{array}
\right.
}
\end{equation}
\end{theorem}

\begin{proof}
By definition, we have 
$$
\qcvxB_Q(\theta:\theta')=\lim_{\alpha\rightarrow 1^-} \frac{1}{\alpha(1-\alpha)} \left(\max\{Q(\theta),Q(\theta')\}-Q((\theta\theta')_\alpha)\right).
$$
Applying a first-order Taylor expansion to $Q\left((\theta\theta')_\alpha\right)$, we get 
\begin{equation}
Q\left((\theta\theta')_\alpha)\right)\simeq_{\alpha\rightarrow 1} Q(\theta')-(1-\alpha)(\theta-\theta')^\top \nabla Q(\theta').
\end{equation}
Thus we have
\begin{equation}
\qcvxB_Q(\theta:\theta')=\lim_{\alpha\rightarrow 1^-} \frac{1}{\alpha(1-\alpha)} \left(\max\{Q(\theta),Q(\theta')\}-Q(\theta')-(1-\alpha)(\theta-\theta')^\top \nabla Q(\theta')\right).
\end{equation}

Consider the following two cases:
\begin{itemize}
	\item Case $\max\{Q(\theta),Q(\theta')\}=Q(\theta')$: That is, $Q(\theta')\geq Q(\theta)$.
	Then it follows that
\begin{eqnarray}
\qcvxB_Q(\theta:\theta') &=&\lim_{\alpha\rightarrow 1^-} \frac{1}{\alpha(1-\alpha)} \left(-(1-\alpha)(\theta-\theta')^\top \nabla Q(\theta')\right),\\
 &=&-(\theta-\theta')^\top \nabla Q(\theta').
\end{eqnarray}

\item Case $\max\{Q(\theta),Q(\theta')\}=Q(\theta)$: That is, $Q(\theta)\geq Q(\theta')$.
Then we have
$$
\qcvxB_Q(\theta:\theta')=\lim_{\alpha\rightarrow 1^-} \frac{1}{\alpha(1-\alpha)} \left(Q(\theta)-Q(\theta')-(1-\alpha)(\theta-\theta')^\top \nabla Q(\theta')\right).
$$
We have $\lim_{\alpha\rightarrow 1^-} Q(\theta)-Q(\theta')-(1-\alpha)(\theta-\theta')^\top \nabla Q(\theta')=Q(\theta)-Q(\theta')=\Delta_Q(\theta:\theta')$ that is finite and different from $0$ when $\theta\not=\theta'$, and therefore $\lim_{\alpha\rightarrow 1^-} \frac{1}{\alpha(1-\alpha)}\Delta_Q(\theta:\theta')=+\infty$.

\end{itemize}

Let us now prove the axiom of non-negativity and disprove the law of the indiscernibles at inflection points for the  quasiconvex Bregman pseudo-divergences.

\begin{figure}
\centering
\includegraphics[width=0.5\textwidth]{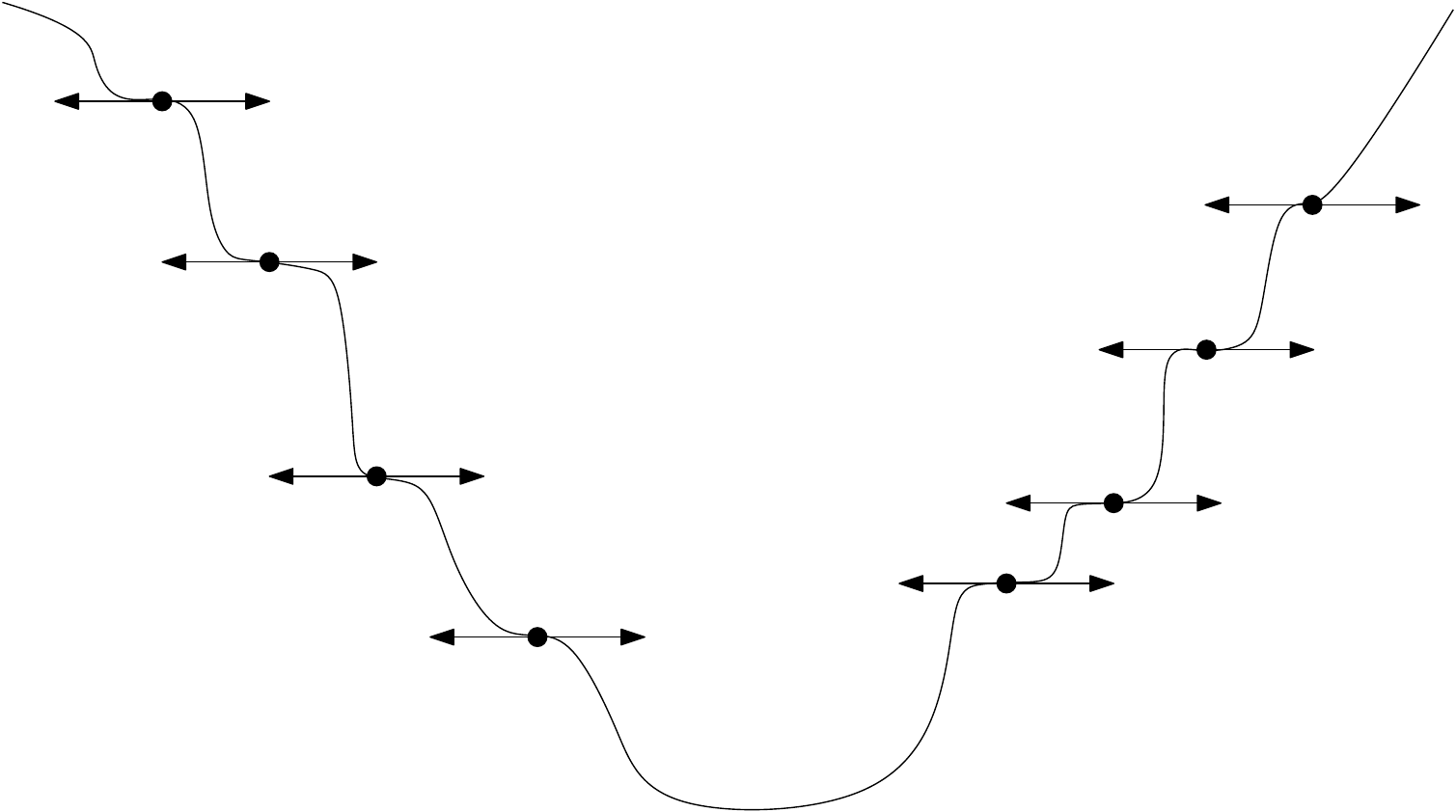}
\caption{An example of a strictly quasiconvex function $Q$ with (countably) many inflection points (at locations $\theta_i$'s) for which the derivative vanishes $Q'(\theta_i)=0$ and the second derivative $Q''$ changes sign at the $\theta_i$'s. \label{fig:inflectionpoint}  }
\end{figure}

\begin{itemize}
	\item Law of the indiscernibles:
Clearly, $\qcvxB_Q(\theta:\theta)=0$ for all $\theta\in\Theta$. 
So consider $\theta\not=\theta'$, and $\qcvxB_Q(\theta:\theta')=-\nabla Q(\theta')^\top(\theta-\theta')=0$ for $Q(\theta')\geq Q(\theta)$.
It is enough to consider the 1D case, by considering the divergence restricted to the line passing through $\theta$ and $\theta'$ intersected by the domain $\Theta$.
We may have countably many inflection points $\theta'$ for which $Q'(\theta')=0$.
At those inflection points, we may find $\theta\not=\theta'$ such that $\qcvxB_Q(\theta:\theta')=0$.
Thus the quasiconvex Bregman divergence {\em does not} satisfy the law of the indiscernibles.
Figure~\ref{fig:inflectionpoint} displays an example of such a quasiconvex function with a few  inflection points.

For example, consider the strictly quasiconvex generator $Q(x) = x^3$, with $\theta < 0$ and $\theta' = 0$.
We have:
 \begin{equation}
  \qcvxJ^\alpha_Q(\theta: \theta') = \mathrm{max}\{Q(\theta), Q(\theta')\}
   - Q((1 - \alpha) \theta + \alpha \theta') =  - (1-\alpha)^3
   \theta^3 > 0.
   \end{equation}

Defining the corresponding quasiconvex Bregman divergence by taking the limit of scaled quasiconvex Jensen divergence yields

   \begin{equation}
 \qcvxB_Q  \lim_{\alpha \to 1} \frac{1}{\alpha(1-\alpha)}
 \qcvxJ^\alpha_Q(\theta: \theta') = \lim_{\alpha \to 1^-} 
- \frac{(1 - \alpha)^2}{\alpha} \theta^3 = 0. 
 \end{equation}

Thus the quasiconvex Bregman divergence is only a pseudo-divergence at countably many inflection points.
Section~\ref{sec:deltaaveraged} will overcome this problem by introducing the $\delta$-averaged quasiconvex Bregman divergence.

 \item Non-negativity follows from a classic theorem of quasiconvex analysis  which reports a first-order condition for a function to be quasiconvex\footnote{By analogy to a classic second-order condition for a strictly convex and differentiable function $F$  to be convex: To have its Hessian $\nabla^2$ positive-definite (Alexandrov's theorem). Similarly, the first-order condition for convexity of a function states that a differentiable function $F$ with convex domain is convex iff. $F(\theta)\geq F(\theta')+(\theta-\theta')^\top \nabla F(\theta')$ from which we recover the Bregman divergence: $B_F(\theta:\theta')=F(\theta)-F(\theta')-(\theta-\theta')^\top \nabla F(\theta')\geq 0$. }:
A $C^1$ function $Q:\Theta\subset \bbR^D\rightarrow \bbR$ is quasiconvex iff. the following property holds (see Theorem~21.14 of~\cite{MathEco-1994} and \S 3.4.3 of~\cite{Boyd-2004}):
\begin{equation}
Q(\theta')\geq   Q(\theta) \Rightarrow \nabla Q(\theta') (\theta-\theta')\leq 0.
\end{equation} 
That is equivalent to $\nabla Q(\theta')^\top(\theta-\theta')\leq 0$ or $\qcvxB_Q(\theta:\theta')=-\nabla Q(\theta')^\top(\theta-\theta')\geq 0$.

Notice that when $Q=F$ is strictly convex and differentiable, then the property also follows from the non-negativity of the corresponding Bregman divergence $B_F(\theta:\theta')\geq 0$ and $F(\theta')\geq F(\theta)$:
\begin{eqnarray}
&& F(\theta)-F(\theta')-(\theta-\theta')^\top \nabla F(\theta')\geq 0,\\
&&\underbrace{-(\theta-\theta')^\top \nabla F(\theta')}_{\qcvxB_F(\theta:\theta')} \geq F(\theta')-F(\theta)\geq 0.
\end{eqnarray}

\end{itemize}

\end{proof}

Notice that $-(\theta-\theta')^\top \nabla Q(\theta')=(\theta'-\theta)^\top \nabla Q(\theta')\geq 0$ when $Q(\theta)\leq Q(\theta')$.
Figure~\ref{fig:qb} illustrates the quasiconvex Bregman divergence for a strictly quasiconvex generator which is strictly concave and has no inflection point.

An interesting property is that if $\qcvxB_Q(\theta:\theta') < \infty$ for $\theta\not=\theta'$ then necessarily $\qcvxB_Q(\theta':\theta) = \infty$, and vice-versa  (when both parameters are not at inflection points).
The forward $\qcvxB_Q$ and reverse $\qcvxB_Q^r$ quasiconvex Bregman pseudo-divergences are both finite only when $Q(\theta)=Q(\theta')$ and
then we have $\qcvxB_Q(\theta:\theta)=0$ or when one parameter is an inflection point.

Moreover, we have the following decomposition for a quasiconvex function $Q\in\calQ$:

\begin{equation}
\eB_Q(\theta:\theta') = Q(\theta)-Q(\theta')+\qcvxB_Q(\theta:\theta'),
\end{equation}
when $Q(\theta)\leq Q(\theta')$,
where $\eB_Q$ stands for the {\em extended Bregman divergence}, i.e., the Bregman divergence extended to a quasiconvex generator.

\begin{figure}%
\centering
\includegraphics[width=0.6\columnwidth]{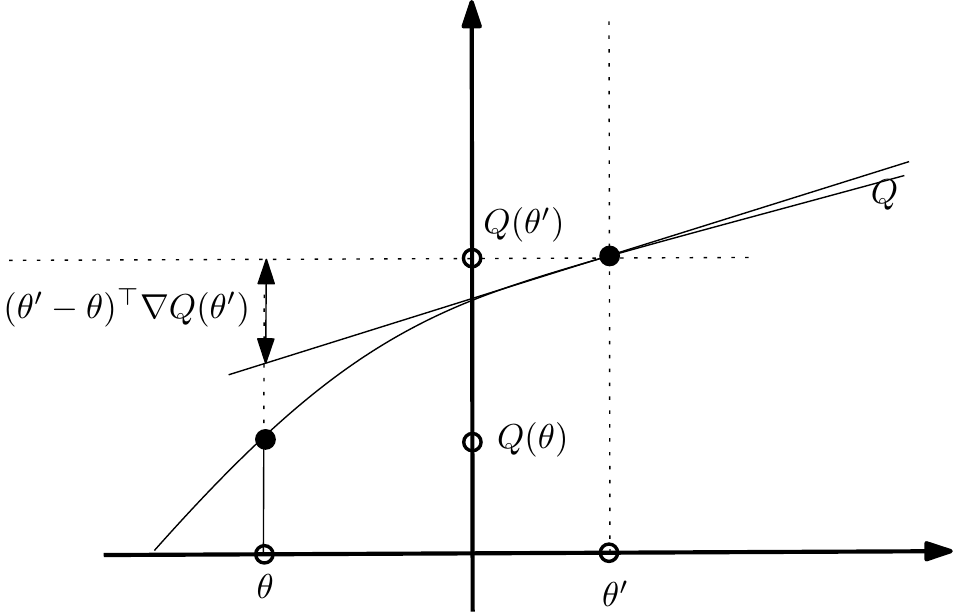}%
\caption{Illustration of the quasiconvex Bregman divergence for a  strictly quasilinear function $Q$ chosen to be concave (e.g. logarithmic type).}%
\label{fig:qb}%
\end{figure}

\begin{remark}[Separability/non-separability of generators and divergences]
When the  $D$-dimensional generator $Q$ is {\em separable}, i.e., $Q(\theta)=\sum_{i=1}^D Q_i(\theta_i)$ where $\theta=(\theta_1,\ldots,\theta_D)$ and the $Q_i$'s are differentiable and quasiconvex univariate functions, the quasiconvex Bregman divergence rewrites as 
\begin{equation}
\qcvxB_Q(\theta:\theta') =\left\{
\begin{array}{ll}
- \sum_{i=1}^D (\theta_i-\theta'_i)  Q_i'(\theta'_i) & \mbox{if $Q(\theta)\leq Q(\theta')$}\\
+\infty & \mbox{otherwise ($Q(\theta)> Q(\theta')$)}.
\end{array}
\right.
\end{equation}

Notice that the condition for the quasiconvex Bregman divergence to be infinite is $Q(\theta)> Q(\theta')$, and not that there exists one index $i\in\{1,\ldots, D\}$ such that $Q_i(\theta_i)> Q_i(\theta'_i)$.
Thus, we have $\qcvxB_Q(\theta:\theta')\not = \sum_{i=1}^D\qcvxB_{Q_i}(\theta_i:\theta'_i)$.
This is to contrast with Bregman divergences for which the separability of the generator $F(\theta)=\sum_{i=1}^D F_i(\theta_i)$ yields the separability of the divergence: $B_F(\theta:\theta')=\sum_{i=1}^D B_{F_i}(\theta_i:\theta'_i)$.
\end{remark}

%%%%%
\subsection{The   $\delta$-averaged quasiconvex Bregman divergence\label{sec:deltaaveraged}}
%%%%%

We shall overcome the problem of indiscernability for quasiconvex Bregman
pseudo-divergences:
\begin{equation}
 \qcvxB_Q(\theta: \theta') = (\theta' - \theta)  Q'(\theta') \quad
 \mathrm{for} \quad
  Q(\theta')  \geq Q(\theta).
\end{equation}

Since the number of inflection points is at most countable for a
strictly quasiconvex generator $Q$, the function $\theta \mapsto \qcvxB_Q(\theta: \theta')$ can
only be identically zero on a set of null measure.
We propose to 
integrate over a  neighborhood of the parameters to obtain a strictly
positive divergence when $\theta' \neq \theta$. 

Given a prescribed parameter
$\delta \neq 0$, we introduce the {\em $\delta$-averaged quasiconvex Bregman divergence}
$\qcvxB_Q^\delta$ via the following definition:
\begin{equation}
\qcvxB_Q^\delta(\theta, \theta') \eqdef \frac{1}{\delta}\int_0^\delta \qcvxB_Q(\theta +u: \theta'+u) \du.
\end{equation}

Choosing $\delta$ to be a strictly positive multiple of $\theta' - \theta$
ensures that this integral is always finite since $Q(\theta' +u) \geq
Q(\theta +u)$ for $u \in I(0, \delta)$, where $I(a, b):= \left\{ta +
(1-t)b, \quad t \in ]0,1[\right\}$ denotes the interval with endpoints $a$
and $b$. 

We now prove this claim. For all $u \in I(0, \delta)$, we have  $\theta'
\in I(\theta, \theta' +u)$ so that
\[Q(\theta') < \max\left\{ Q(\theta), Q(\theta' + u)\right\} = Q(\theta' +
u) \quad \mathrm{since} \quad Q(\theta) \leq Q(\theta').\]
Similarly, $\theta + u \in I(\theta,\theta')$ or  $\theta +
u \in I(\theta',\theta'+u)$.
In the first case, if $\theta + u \in I(\theta,\theta')$ we have
\[Q(\theta + u) < \max\left\{ Q(\theta), Q(\theta' +
  u)\right\} \leq Q(\theta' + u).\] In the second case,  $\theta +
u \in I(\theta',\theta'+u)$, and we obtain 
\[ Q(\theta + u) < \max\left\{ Q(\theta'), Q(\theta' +
    u)\right\} \leq Q(\theta' + u),\] proving the claim.

By construction, this $\delta$-averaged quasiconvex Bregman divergence now satisfies
the law of the indiscernables.

When $Q$ is differentiable, we obtain:
\begin{equation}
  \label{eq:1}
\qcvxB_Q^\delta(\theta, \theta') \eqdef \frac{1}{\delta}\int_0^\delta (\theta' -
\theta)Q'(\theta'+ u) du = (\theta' - \theta)\left(\frac{Q(\theta' +
  \delta) - Q(\theta')}{\delta}\right).
\end{equation}

We note that the rhs. of (\ref{eq:1}) can also serve as the definition of
the $\qcvxB^\delta_Q$ divergences, even when the  strictly quasiconvex function $Q$
is not differentiable. This motivates us to introduce the next
definition, where we now denote by $\delta > 0$ the positive ratio between
$\delta$ and $\theta' - \theta$ of the preceding section.

\begin{definition}[$\delta$-averaged quasiconvex Bregman divergence]
For a prescribed $\delta > 0$ and a strictly quasiconvex generator $Q$ not necessarily differentiable, the 
$\delta$-averaged quasiconvex Bregman divergence is defined by
\begin{equation}\label{eq:deltaaveraged}
\boxed{
\qcvxB_Q^\delta(\theta, \theta') \eqdef 
\left\{ 
\begin{array}{ll}
\frac{1}{\delta} \left(Q\left( \theta' +  \delta (\theta' - \theta)\right) -
  Q(\theta')\right) & \mbox{if $Q(\theta') \geq Q(\theta)$}\\
+\infty & \mbox{otherwise}
\end{array}
\right.
}
\end{equation}
\end{definition}

Let us report some examples of $\delta$-averaged quasiconvex Bregman divergences:
\begin{itemize}
\item $Q(x) = x$.
   \[\qcvxB_Q(\theta: \theta') = \frac{(1 + \delta)\theta' - \delta
       \theta - \theta'}{\delta} = \theta' - \theta, \]
   when $ \theta' \geq \theta$, or $+\infty$ otherwise.
   \item $Q(x) = x^2$.
   \[\qcvxB_Q(\theta: \theta') = 2 \theta'(\theta' - \theta) + \delta
     (\theta' - \theta)^2,\]
   when $ |\theta'| \geq |\theta|$,  or $+\infty$ otherwise.
      \item $Q(x) = x^3$.
   \[\qcvxB_Q(\theta: \theta') = 3 \theta'^2 (\theta' - \theta) + 3
     \theta' \delta (\theta' - \theta)^2 + \delta^2(\theta' -\theta)^3,\]
   when $ \theta' \geq \theta$,  or $+\infty$ otherwise.
   At the inflection point $\theta' = 0$, we now have
   \[ \qcvxB_Q(\theta: \theta') = -\delta^2 \theta^3 > 0 \quad \forall \theta
     < 0.\]
\end{itemize}

%%%
\subsection{Quasiconvex Bregman divergences as limit cases of power mean Bregman divergences}\label{sec:PowerMeanB}
%%%
For sake of simplicity, consider scalar divergences below.
In~\cite{CC-BD-2017}, the {\em $(M,N)$-Bregman divergence} is defined as the limit case:
\begin{equation}\label{eq:BCCD}
B_{F}^{M,N}(p:q) = \lim_{\alpha\rightarrow 1^-}   \frac{1}{\alpha(1-\alpha)}J_{F,\alpha}^{M,N}(p:q) = \lim_{\alpha\rightarrow 1^-}  \frac{1}{\alpha(1-\alpha)} \left( N_\alpha(F(p),F(q)))-F(M_\alpha(p,q)) \right).
\end{equation}

In particular, the univariate {\em power mean Bregman divergences} are obtained by taking the power means, yielding the following formula:

\begin{equation}
B^{\delta_1,\delta_2}_F(p:q) = \frac{F^{\delta_2}(p)-F^{\delta_2}(q)}{\delta_2 F^{\delta_2-1}(q)} - \frac{p^{\delta_1}-q^{\delta_1}}{\delta_1 q^{\delta_1-1}}F'(q).
\end{equation}

Let $\delta_2=r$ and $\delta_1=1$. Then we get the subfamily of $r$-power Bregman divergences:

\begin{eqnarray}
B^{r}_F(\theta:\theta') &=& \frac{F^{r}(\theta)-F^{r}(\theta')}{r F^{r-1}(\theta')}-(\theta-\theta')F'(\theta'),\\
&=&  = \frac{F^{r}(\theta)}{r F^{r-1}(\theta')}-  \frac{F(\theta')}{r}  -(\theta-\theta')F'(\theta').\label{eq:expan}
\end{eqnarray}

In Eq.~\ref{eq:expan}, when $F(\theta)>F(\theta')$ then we have $\lim_{r\rightarrow\infty} B^{r}_F(\theta:\theta')=\infty$
since $\left(\frac{F^r(\theta)}{F^{r-1}(\theta')}\right)$ diverges.
Otherwise $\qcvxB_F(\theta:\theta')=\lim_{r\rightarrow\infty} B^{r}_F(\theta:\theta')=-(\theta-\theta')F'(\theta')$ since $\lim_{r\rightarrow}  \frac{F(\theta')}{r} =0$ 
(because $|F(\theta')|<\infty$).

When $r\rightarrow\infty$, the power mean operator $P_r$ tends to the maximum operator: $\lim_{r\rightarrow\infty} P_r(a,b)=\max\{a,b\}$, and
the $(A,P_\delta)$-Bregman divergence tends to the quasiconvex Bregman pseudo-divergence.

\subsection{Some illustrating examples of quasiconvex Bregman divergences}

We concisely report two univariate quasiconvex scalar Bregman divergences:

\begin{itemize}

\item For $Q(\theta) = \theta$ with $\theta\in\bbR$, we have 
$$
\qcvxJ_Q^\alpha(\theta:\theta') = \max\{\theta,\theta'\} - (1-\alpha)\theta+\alpha\theta'.
$$

We consider the two cases for calculating the limit $\qcvxB_Q(\theta:\theta') =
 \lim_{\alpha\rightarrow 1^-} \frac{1}{\alpha(1-\alpha)}\qcvxJ_Q^\alpha(\theta:\theta')$:
\begin{itemize}

\item When $\theta'\geq\theta$: 
$$
\lim_{\alpha\rightarrow 1^-} \frac{1}{\alpha(1-\alpha)}\qcvxJ_Q^\alpha(\theta:\theta')=\lim_{\alpha\rightarrow 1^-} \frac{1}{\alpha(1-\alpha)} 
(-(1-\alpha)\theta+(1-\alpha)\theta')=\theta'-\theta\geq 0.
$$

\item When $\theta > \theta'$:
$$
\lim_{\alpha\rightarrow 1^-} \frac{1}{\alpha(1-\alpha)}\qcvxJ_Q^\alpha(\theta:\theta')=\lim_{\alpha\rightarrow 1^-} \frac{1}{\alpha(1-\alpha)} 
(\theta-(1-\alpha)\theta-\alpha\theta')=\lim_{\alpha\rightarrow 1^-}  \frac{1}{1-\alpha}(\theta-\theta') = +\infty. 
$$

\end{itemize}
Thus we have the following quasiconvex Bregman divergence:
$\qcvxB_Q(\theta:\theta')=\theta' - \theta$  for $\theta' \geq \theta$ and $+\infty$ when $\theta' < \theta$.

\item When $Q(\theta)=\log\theta$, we have $Q'(\theta)=\frac{1}{\theta}$ and
$\qcvxB_Q(\theta:\theta')=1-\frac{\theta}{\theta'}$ for $\log\theta' \geq \log\theta$ (i.e. $\theta'\geq \theta$) and 
$+\infty$ when $\theta' < \theta$.

\item For $Q(\theta)=\sqrt{\theta}$ and $\theta\in\Theta=(0,\infty)$, we have $Q'(\theta)=\frac{1}{2\sqrt{\theta}}$ and
$\qcvxB_Q(\theta:\theta')= \frac{1}{2} \left( \sqrt{\theta'} - \frac{\theta}{\sqrt{\theta'}}\right)$ for  $\sqrt{\theta'}\geq\sqrt{\theta}$ (i.e., $\theta' \geq \theta$), and $+\infty$ when $\theta' < \theta$.

\end{itemize}

%%%
\section{Statistical divergences, parametric families of distributions and equivalent parameter divergences}\label{sec:StatParamDiv}
%%%

Consider a probability space $(\calX,\calF,\mu)$ with $\calX$, $\calF$, and $\mu$ denoting the sample space, the $\sigma$-algebra and the positive measure, respectively.
The most celebrated {\em statistical divergence}   between two densities $p_{\theta}\ll\mu$ and $p_{\theta'}\ll \mu$ absolutely continuous with respect to a measure $\mu$ is the Kullback-Leibler (KL) divergence (also called {\em relative entropy}~\cite{CT-2012}), defined by:
\begin{equation}
\KL[p:q]=  \left\{
\begin{array}{ll}
\int_{x\in\calX} p(x)\log\frac{p(x)}{q(x)}\dmu(x), & \supp(p)\subset\supp(q),\\
+\infty, & \supp(p)\not\subset\supp(q).
\end{array}
\right.,
\end{equation}
where    $\supp(p)=\{x\in\bbR \ :\ p(x)>0\}$ denotes the {\em support} of a distribution $p(x)$, and $\log\frac{0}{0}=0$ by convention.
Thus the KL divergence is said unbounded in general.\footnote{The Jensen-Shannon divergence~\cite{JS-2019} is a particular symmetrization of the KL divergence which is always bounded, and may accept densities with different supports.}

In general, a statistical divergence between densities belonging to the same parametric family $\calP=\{p_{\theta}\}_\theta$ of mutually absolutely continuous densities is equivalent to a corresponding {\em parameter divergence} $B$:
\begin{equation}
B(\theta:\theta')\eqdef D[p_{\theta}:p_{\theta'}].
\end{equation}

For example, when $\calP=\{p_\theta(x)=\exp(x^\top\theta-F(\theta))\dmu(x)\}_\theta$ is an exponential family~\cite{MLE-EF-1997,EF-Barndorff-2014,BD-2005} on a probability space $(\calX,\calF,\mu)$, then the Kullback-Leibler divergence between two densities of the exponential family (e.g., two Gaussians distributions belonging to the Gaussian exponential family) amount to a {\em reverse} Bregman divergence~\cite{BD-2005} for the Bregman generator set to the cumulant function $F(\theta)=\log\int \exp(x^\top\theta)\dmu(x)$:

\begin{equation}
\KL[p_{\theta}:p_{\theta'}]={B}(\theta:\theta')={B_F}^r(\theta:\theta')=B_F(\theta':\theta).
\end{equation}

Banerjee et al.~\cite{BD-2005} proved a {\em bijection} between regular natural exponential families and so-called {\em regular} Bregman divergences.
Note that since the Csisz\'ar's $f$-divergence~\cite{fdiv-AliSilvey-1966,IG-2016} (including the KL divergence) is invariant to one-to-one smooth mapping $m(x)$ of the sample space $x$, the same Bregman divergence equivalent to the KL divergence can be obtained for different exponential families where $y=m(x)$.
For example, the KL divergence between two normal distributions or two ``equivalent'' log-normal distributions is the same (using the mapping $y=\log x$).
This can be also noticed by the matching of their cumulant function: $F_{\mathrm{normal}}(\theta)=F_{\mathrm{lognormal}}(\theta)$.

Quasiconvex Bregman divergences have the interesting property to be  finite for one orientation and  infinite for the other orientation.
Thus to find an example of parametric family of distributions which the KL divergence amount to a quasiconvex Bregman divergence, we shall consider parametric distributions with {\em nested supports} (or {\em nested densities}), so that one orientation of the KL divergence will be finite while the other  is will be equal to infinity.

For example, consider the family of univariate uniform densities ($D=1$): 
\begin{equation}
p_\theta(x) = 1_{  0 < x < e^\theta }\ e^{-\theta},
\end{equation}
where $1_A$  denotes the indicator function of $A$.
We have $\supp(p_{\theta'})\subset \supp(p_\theta)$ for $0<\theta'\leq \theta$.
Then we have

\begin{equation}
\KL[p_\theta:p_{\theta'}]=  \left\{
\begin{array}{ll}
\theta' - \theta = \qcvxB_Q(\theta: \theta') &  0<\theta\leq\theta',\\
+\infty & \theta'>\theta.
\end{array}
\right.,
\end{equation}
for $Q(\omega) = \omega$.

Notice that the family $\calP=\{p_\theta\}$ is not an exponential family since the family has not a fixed support.
A truncated exponential family with fixed truncation parameters yields an exponential family which may neither be regular nor steep (e.g., the singly truncated normal distributions~\cite{EFnonsteep-1994}).

Now, consider the parametric family $\{q_\theta\}_\theta$ of nested densities:
\begin{equation}
q_\theta(x) = 1_{0 < x < e^\theta}\alpha   \frac{x^{\alpha - 1}}{e^{\theta \alpha}},
\end{equation}
for a prescribed $\alpha> 1$.
After a short calculation (or using a computer algebra system as reported in Appendix~\ref{app:cas}), we find that
\begin{equation}
\KL[q_\theta:q_{\theta'}]=  \left\{
\begin{array}{ll}
\alpha (\theta' - \theta) = \qcvxB_Q(\theta: \theta') &  \theta'\geq\theta>0,\\
+\infty & \theta'<\theta.
\end{array}
\right.,
\end{equation}
for $Q(\omega) = \omega$.
Thus we have built several parametric families of nested densities that up to a scaling factor yields the same quasiconvex Bregman divergence.

For parametric densities belonging to the same exponential family, it is known that the Bhattacharrya distance amount to a Jensen divergence~\cite{BR-2011}.
For an exponential family $p_\theta(x)=\exp(\theta^\top x-F(\theta))\dmu(x)$ with cumulant function $F$, the cross-entropy between two densities~\cite{EF-2010} is
\begin{equation}
h(p_\theta:p_{\theta'})=\int -p_\theta(x)\log p_{\theta'}(x) \dmu(x) = F(\theta')-(\theta')^\top \nabla F(\theta),
\end{equation}
and the entropy is
\begin{equation}
h(p_\theta)=h(p_\theta:p_{\theta})= F(\theta)-\theta^\top \nabla F(\theta).
\end{equation}

Since $\KL(p_\theta:p_{\theta'})=B_F(\theta':\theta)=F(\theta')-F(\theta)-(\theta'-\theta)^\top\nabla F(\theta)$,
when $F(\theta')\leq F(\theta)$, we have $-(\theta'-\theta)^\top\nabla F(\theta)=\qcvxB_F(\theta':\theta)$, and it follows that

\begin{equation}
\qcvxB_F(\theta':\theta)=\KL(p_\theta:p_{\theta'})+F(\theta)-F(\theta'),\quad F(\theta')\leq F(\theta).
\end{equation} 

The Wasserstein distance between two nested univariate distributions has been studied in~\cite{NestedDensity-2017} with applications to Bayesian statistics to study the influence of the prior distribution in the posterior distribution in the finite sample size setting.

%%%
\section{Conclusion and perspectives}\label{sec:concl}
%%%%

We have introduced   novel families of distortions between vector parameters: The quasiconvex Jensen divergences and the quasiconvex Bregman divergences. We showed that the quasiconvex Jensen divergences measuring the difference gaps of the quasiconvex inequalities  can be interpreted as a $\ell_1$-regularized ordinary Jensen divergence.
We noticed that any quasiconcave Jensen divergence amounts to an equivalent quasiconvex Jensen divergence for the negative generator.
We then derived the quasiconvex Bregman pseudo-divergences as limit cases of scaled and skewed quasiconvex Jensen divergences for  strictly quasiconvex generators. The quasiconvex Bregman pseudo-divergences is a pseudo-divergence only at countably many inflection points of the generators. We thus propose to define the $\delta$-averaged quasiconvex Bregman divergences by integrating the pseudo-divergence over a small neighborhood. This yields a formula (Eq.~\ref{eq:deltaaveraged}) that can be used as the definition of the quasiconvex Bregman divergence even for non-differentiable strictly quasiconvex generators.
We also showed how to derive again the result of the quasiconvex Bregman pseudo-divergences using comparative convexity using the limit case of power means.
A key property of the  quasiconvex  Bregman divergences between distinct elements is that they are necessarily finite on one orientation and infinite for the opposite orientation.
Finally, we showed how some of these quasiconvex Bregman divergences can be obtained from the Kullback-Leibler divergence between densities belonging to the same parametric family of distributions with nested support.
We can retrieve the Bregman pseudo-divergences and quasiconvex Bregman pseudo-divergences from first-order convexity and quasiconvexity conditions, as illustrated in Table~\ref{tab:comparison}.
Additional conditions on the generators ensure that the pseudo-divergences are divergences and satisfy the law of the indiscernibles 
(i.e., strict convexity and differentiability for Bregman divergences and strict quasiconvexity without inflection points for the quasiconvex Bregman divergences).

In future work, we shall consider applications of these novel divergences like clustering: 
We note that the  generic $k$-means++ probabilistic seeding analysis reported in~\cite{tJ-2015} does not apply because of the forward/reverse infinite property of these quasiconvex Bregman divergences.
We may consider discrete $k$-means, $k$-center (with the minimum enclosing ball obtained from quasiconvex programming~\cite{QcvxProgramming-2005,QuasiconvexGeoOptRec-2007,qvexopt-2015,QcvxProgramming-2019} when $k=1$), and  quasiconvex Bregman hierarchical clustering~\cite{BregHClust-2012}.

\begin{table}
{\small
\begin{tabular}{|c|l||l|}\hline
& First-order  condition & Pseudo-divergence/condition for divergence \\ \hline\hline
Convexity & $F(\theta)\geq F(\theta')+(\theta-\theta')^\top \nabla F(\theta')$ & $B_F(\theta:\theta')=F(\theta)-F(\theta')+(\theta-\theta')^\top \nabla F(\theta')$ \\
 of $F$ &   \multicolumn{2}{c|}{  Divergence when $F$ strictly convex and differentiable}\\ \hline
Quasiconvexity & $Q(\theta)\leq Q(\theta') \Rightarrow (\theta-\theta')^\top \nabla Q(\theta')\leq 0$ & 
$\left\{
\begin{array}{ll}
-(\theta-\theta')^\top \nabla Q(\theta') & \mbox{if $Q(\theta)\leq Q(\theta')$}\\
+\infty & \mbox{otherwise}.
\end{array}
\right.$
\\
 of $Q$ &   \multicolumn{2}{c|}{  Divergence when $Q$ strictly quasiconvex with no inflection point}\\ \hline
\end{tabular}
}

\caption{Bregman divergence and Bregman quasidivergence with their relationship to first-order convexity and quasiconvexity.\label{tab:comparison}}
\end{table}

\appendix

%%%%
\section{Calculations using a computer algebra system\label{app:cas}}
%%%%
Using the computer algebra system {\sc Maxima}\footnote{Freely downloadable at \url{http://maxima.sourceforge.net/}}, we report the calculation of the KL divergence for nested densities.

\begin{verbatim}
assume(alpha>1);
assume(theta>0);
p(x,theta):=alpha*(x**(alpha-1))/(exp(theta*alpha));
integrate(p(x,theta),x,0,exp(theta)); 
assume(thetap>theta);
/* KL divergence */
integrate(p(x,theta)*log(p(x,theta)/p(x,thetap)),x,0,exp(theta));
\end{verbatim}

\bibliographystyle{plain}
\bibliography{QuasiConvexBregmanBIB}
\end{document}